\documentclass[sigconf]{acmart}  
\usepackage{graphicx}
\usepackage{textcomp}
\usepackage{xcolor}
\usepackage{bbm}

\usepackage{amsmath}
\usepackage{amsthm}
\usepackage{float}
\usepackage{algorithm}
\usepackage[noend]{algpseudocode}
\usepackage{booktabs} 

\def\BState{\State\hskip-\ALG@thistlm}

\def\BibTeX{{\rm B\kern-.05em{\sc i\kern-.025em b}\kern-.08em
    T\kern-.1667em\lower.7ex\hbox{E}\kern-.125emX}}

\acmArticle{4}
\begin{document}


\acmDOI{}
\acmISBN{}
\acmPrice{}



\title{Tracking Where Events Take Place: Reverse Spatial Term Queries on Streaming Data}
\author{Sara Farazi}
\affiliation{
  \institution{University of Alberta}
  \city{Edmonton}
  \state{Canada}
}
\email{farazi@ualberta.ca}

\author{Davood Rafiei}
\affiliation{
  \institution{University of Alberta}
  \city{Edmonton}
  \state{Canada}
}
\email{drafiei@ualberta.ca}

\begin{abstract}
A large volume of content generated by online users is geo-tagged and this provides a rich source for querying in various location-based services.
An important class of queries within such services involves the association between content and locations.
In this paper, we study two types of queries on streaming geo-tagged data:
1) \textit{Top-k reverse frequent spatial queries}, where given a term, the goal is to find top $K$ locations where the term is frequent, and 2) \textit{Term frequency spatial queries}, which is finding the expected frequency of a term in a given location. 
To efficiently support these queries in a streaming setting,
we model terms as events and explore a probabilistic model of geographical distribution that allows us to estimate the frequency of terms in locations that are not kept in a stream sketch or summary. 
We study the back-and-forth relationship between the efficiency of queries, the efficiency of updates and the accuracy of the results and identify some sweet spots where both efficient and effective algorithms can be developed. We demonstrate that our method can be extended to support multi-term queries. To evaluate the efficiency of our algorithms, we conduct experiments on a relatively large collection of both geo-tagged tweets and geo-tagged Flickr photos.
The evaluation reveals that our proposed method achieves a high accuracy when only a limited amount of memory is given. Also the query time is improved, compared to a recent baseline, by 2-3 orders of magnitude without much loss in accuracy and that the update time can further be improved by at least an order of magnitude under some term distributions or update strategies.
\end{abstract}

\maketitle

\section{Introduction} \label{section:introduction}

The share of web traffic generated by mobile devices has been on the rise for the past 10 years and it has passed the 50\% mark in the last four years~\cite{mobileShare2018}.
A wide range of content is generated by users of such devices (e.g., search logs, microblogging posts and comments, photo uploads, measurements, etc.) and this content is often geo-tagged.
On the other hand, many applications can benefit from harnessing the relationship between content and locations. For example, if the content is treated as events, the relationship may indicate where the events take place.

In a typical setting, where geo-tagged events are streaming in and location-aware queries are processed on the fly, the number of events being recorded and queries being processed can be overwhelming. For example, consider Twitter where 500 million tweets are streaming in every day and the data is queried by 326 million active users of the network~\cite{twitterStat}.
In such settings, the efficiency of maintaining statistics about events while the data is streaming in is crucial, and
any overhead due to querying has to be minimized.
The problem of detecting top-k frequent or trending terms in a spatio-temporal region has been studied and efficient algorithms have been developed~\cite{jensen:topk, garnet, geotrend, venus}. 
In this paper, we study the attachment between events and locations in a
reverse direction, where given a term denoting an event, we want to find locations where the term is frequent.
In particular, we study the problem of efficiently supporting two types of spatial queries: (1) Reverse Frequent Spatial (RFS) queries, where top-K most frequent locations for a given term are retrieved, and (2) Term Frequency Spatial (TFS) queries, where given a term and a geographical location, the frequency of the term in that location is estimated.
Examples of RFS queries are finding top locations where an entity such as ``Trump'' or an event such as a baseball match is trending, or finding top locations where flu symptoms have been frequently reported by users. Examples of TFS queries are estimating the number of users from a specific location who mention ``flu'' or ``vaccine'' in their tweets, or estimating the sales volume of a product in a region, for example, based on the number of tweets from the region that mention the product. 

The queries studied in this paper arise in many domains where the association between content and geography is important and this association changes over time. For static or near-static associations (e.g. the spread of restaurants and food types in different locations), the information may be materialized or organized in advance for more efficient access. However, this is not feasible when the content as well as query results are regularly changing. Some specific use cases are querying for the geographical distribution of natural events (e.g. storm, earthquake) in a stream of tweets or hashtags, detecting locations where a product, a political party or a politician name is trending, and finding locations where certain social, cultural or sport event may be happening. The diversity of examples indicates the wide-ranging applications and domains where reverse spatial queries arise or may be useful. For many of these use-cases, the term distribution can change over time with the input, and so is the answer. Hence, an accurate answer that is costly (in terms of the time it takes or the resources it uses) may not be desirable, and an approximate answer with some bounded error may be preferred. Also when the query output is fed into a map visualization tool (which is common in these settings), some degree of error can be tolerated as well.

A typical approach for supporting the aforementioned queries is to use an inverted index that efficiently returns the set of locations where a given event is reported~\cite{almaslukh2018evaluating,zhao2017monochromatic}. One may divide the world map into grid cells and only store the cells, instead of exact locations, in the posting lists, hence cutting the posting list size and the query processing time. However, this approach does not scale well, in terms of space and query and update time, to a large number of events and queries.
As an example, consider the set of tweets generated in the US, and suppose we are interested in locations that cite a particular event. For example, suppose we are interested in locations that mention the term ``hollywood.'' Having detailed (location,event) data, as shown in Figure~\ref{fig:hollywoodmap} for the term ``hollywood'' in December 2017, one can answer both RFS queries such as ``top few locations where the term or event is cited'' and TFS queries such as ``the number of tweets from San Jose that mention the term.''
However, storing such detailed map for a large number of events or terms of possible interest and regularly maintaining the data with incoming tweets is both difficult and resource intensive.
Also maintaining detailed maps does not help with query processing, and it can in fact increase the query runtime. An important problem in this context is 
understanding some of the trade offs between the efficiency of queries, the efficiency of updates and the storage cost for the streaming input.
\begin{figure}[htb]
    \centering
    \includegraphics[width=3.0in]{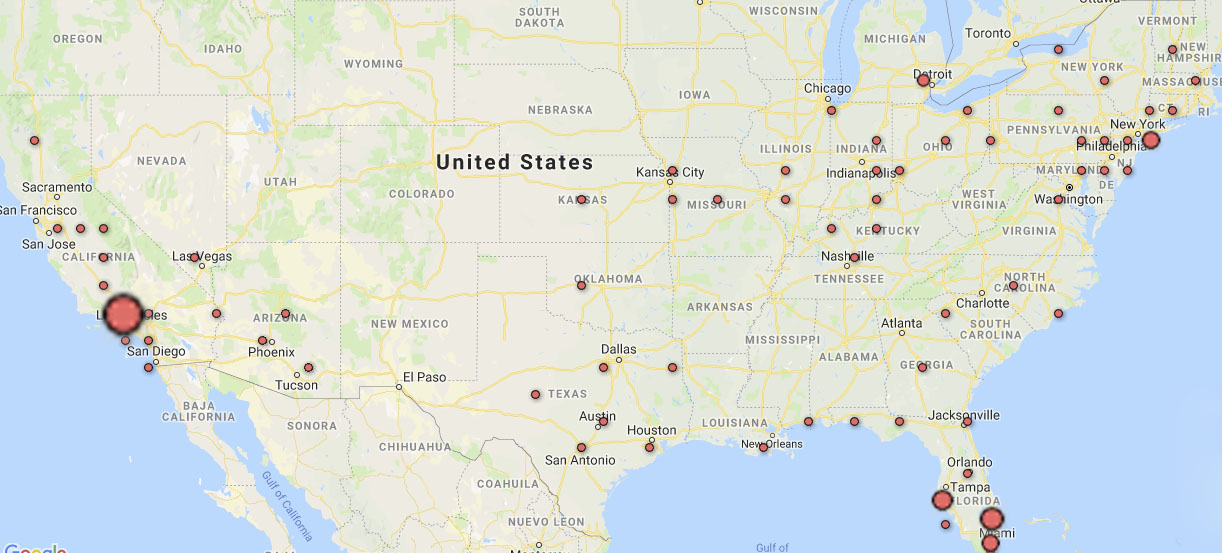}
    \caption{The set of tweets from the US that mention {\it hollywood} in December 2017, projected on the US map.}
    \label{fig:hollywoodmap}
\end{figure}

\subsection{Problem studied and the outline of our approach}
Efficiently evaluating reverse spatial queries on streaming data poses some interesting challenges which have not been addressed in the literature. First, there are both space overhead and update cost for maintaining locations for each term or event. 
The space usage is the product of the number of different terms and locations, 
which can be costly when the number of possible terms is large or the data is kept for all possible locations.
Second, events are not uniformly distributed in all locations, and when an event is frequent in a location, it is more likely to be frequent in the neighbouring locations as well.
Ignoring the  locality relationships
(e.g., containment, proximity, adjacency) can impact both the query processing and update time as well as increasing the space overhead.
Third, assuming that the frequency data is not recorded for all locations (e.g., due to the cost or the overhead), accurately answering TFS queries for locations where the frequency data is not recorded is another challenge.
The literature on frequent item counting in streaming setting can reduce the space overhead and per insert processing cost while introducing some bounded error.
However, the relevant algorithms are very limited when applied to TFS queries; for example, one cannot efficiently estimate the
frequency of a term at a specific location for which the data is not collected.

Instead of maintaining a detailed map of events, one may keep some statistics about the data distribution with
some desirable properties such as bounded error for queries, efficient update time and efficient query processing time. For example,
for the tweets that mention ``hollywood'' in the US, one may keep information about a few hot spots where the event is frequent and some parameters that show how other data points are distributed both around the hot spots and the areas away from the hot spots. One such map is shown in Figure~\ref{fig:hollywood3centers}.
An interesting question is what statistics can be maintained such that (1) TFS and RFS queries can be answered both efficiently, (2) the accuracy of queries, compared to the case when all data points are used, is not much affected, and (3) the statistics can be efficiently updated in a streaming setting.

\begin{figure}[htb]
    \centering
    \includegraphics[width=3.0in]{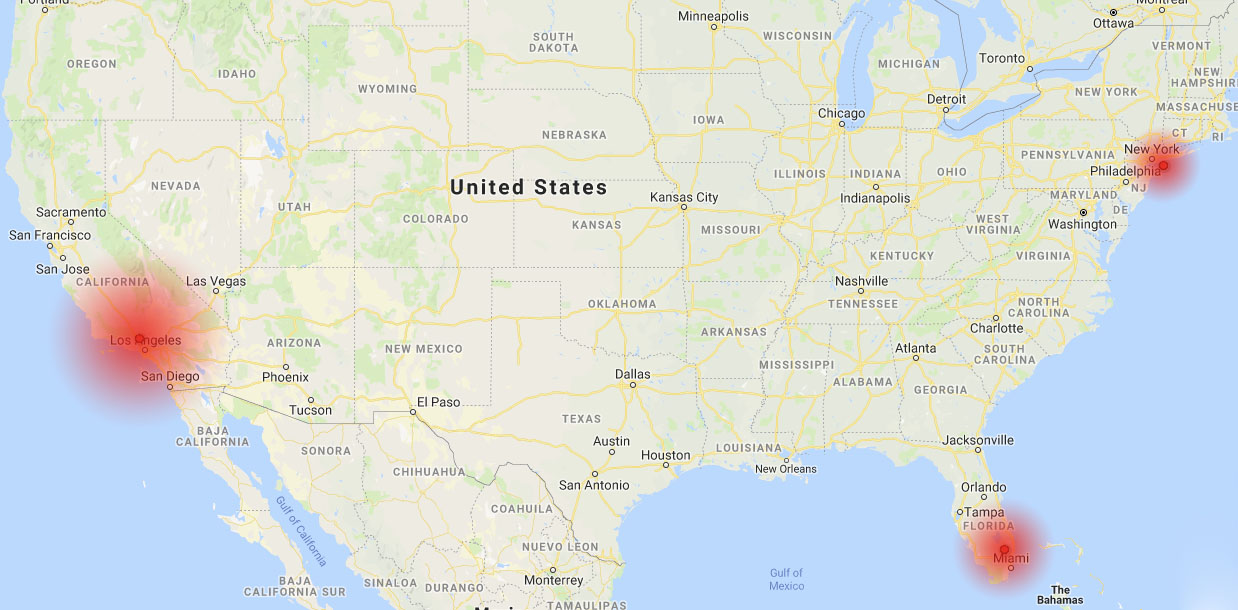}
    \caption{Three hot spots where the term {\it hollywood} is frequently mentioned and the distribution of the event around those hot spots.}
    \label{fig:hollywood3centers}
\end{figure}       

This paper studies the relationship between the efficiency of queries and updates and the accuracy of the results and proposes a new algorithm and a few update strategies that strike a balance between efficiency, accuracy and space usage. 
The proposed algorithms use a counter-based summary~\footnote{Counter-based techniques are more commonly used for point queries because of their relative error guarantees (e.g., \cite{jensen:topk,garnet}).} to maintain top-k frequent locations for each term and a probabilistic spatial distribution model to estimate term frequencies in every location on the map. Hence the frequency can be estimated for locations that are not explicitly kept in the summary.
Using the model, on the other hand, allows us to reduce the space usage of our algorithms to a level that is sufficient to estimate the model parameters.
RFS queries can be directly answered using the summary, whereas for TFS queries, having the total frequency of query terms and the probability from the spatial model allows us find the expected frequency of terms in unsaved locations.

\noindent{\bf Contributions}
The contributions of the paper can be summarized as follows.
First, we introduce a probabilistic model for answering TFS queries, and an efficient algorithm for estimating the model parameters in a streaming setting. 
To the best of our knowledge, this is the first time the model is used in a streaming setting and with only partial information about location frequencies.
Second, we show how the model parameters can be estimated using a counter-based summary and further propose a ring based summary for efficiently answering both RFS and TFS queries.
Third, we develop update strategies to further improve the stream processing time while trading a bit of accuracy in some cases.
Finally, we conduct extensive experiments to evaluate our algorithms and models under different settings and parameter values.

The remainder of the paper is organized as follows:
A counter-based summary augmented with a probabilistic model is presented 
in section~\ref{section:baseline}, and a new algorithm is discussed in Section~\ref{section:ring_summary}. Multi-term queries are discussed in Section~\ref{multi} and the improvements and extensions to our proposed method are presented in Section~\ref{section:possible_improvements}. Our experimental evaluation and our algorithm settings are further discussed in Section~\ref{section:experiments}. Section~\ref{section:related} reviews the related work and Section~\ref{section:conclusion} concludes the paper.

\section{Tsum with a Probabilistic Model}
\label{section:baseline}

Suppose the world map is divided into a uniform grid of pre-determined cell size \cite{grid}, and we associate all events tagged with locations inside a cell to the cell. The size of a cell, given as a parameter, indicates the resolution at which location information is kept.
One may also consider instead the set of locations at a fixed level of a location hierarchy (aka gazetteer). 
The resolution or the level at which the location information is maintained may be set based on the query workload and the level or resolution sought in queries. 
With the frequency of events maintained at a fixed granularity or level of a location hierarchy (e.g., at the level of city, or $1^{\circ} \times 1^{\circ}$ grid cell), the frequency can be computed at higher levels (e.g., state level) by estimating the frequency for all locations that are part of a higher level entity and adding those frequencies. Similarly, the frequency can be computed at lower levels of a hierarchy (e.g., a neighbourhood in a city) under some assumption about the data distribution in a location or grid cell (e.g. uniformity).

In our baseline solution, a stream summary for each term is maintained as a Tsum~\cite{FaraziRafiei19a}, which is a counter-based frequent item counting algorithm to keep track of the top frequent locations for each term.

\noindent\textbf{Tsum structure}
A Tsum consists of a set of locations (denoted by cell Ids) where the term is observed, the frequency of the term in those locations $(f)$ and a frequency error $\Delta$, which is the maximum error in the value of $f$. The summary size (or the number of counters) is set based on the error threshold that can be tolerated.

\begin{lemma}
\label{thm:sizeBound}
Given a stream of locations for a term, one needs a Tsum of size at least $m=\frac{1}{\epsilon}$ to find all $\epsilon$-frequent locations (i.e., locations that appear at least $\epsilon N$ times in a stream of length $N$).
\end{lemma}

Tsum uses the SpaceSaving~\cite{space_saving} algorithm to track top frequent locations. According to this method, each counter keeps record of a cell Id, its frequency, and the error in frequency. The frequency and error values are all zero at the beginning. Each incoming item is recorded in counters until there is no empty counter left. Then if a new incoming item does not exist in the current set of counters, an item with the lowest frequency is selected to be replaced with the new item. The value of the counter is incremented and the previous value of the counter is set as the current error value.

For example, given the stream \texttt{$<$columbus,dallas,dallas,new york,columbus,dallas,chicago$>$} and space to store only $3$ locations, the Tsum method will initially have \texttt{(dallas,3,0),  (columbus, 2,0), (new york,1,0)} where the second and the third entries for each location is its frequency and error respectively. Upon seeing `chicago', since there is no room in the summary, the location with the least frequency is replaced with `chicago' and its error is increased to 1. The new term summary will have \texttt{(dallas,3,0), (columbus,2,0), (chicago,2,1)}.




\noindent
{\bf Error bound}
Tracking the top frequent locations with SpaceSaving gives an error bound on the frequency of a term in each location. 

\begin{lemma}
\label{thm:errorBound}
For a stream window of size $N$ and a Tsum of size $m$, the error $\Delta$ cannot exceed $\frac{N}{m}$.
\end{lemma}

\noindent
{\bf Queries}
Since the top frequent locations for each term are kept in the term summary, RFS queries can be answered by getting the first K counters in Tsum (assuming K is not larger than the summary size). The term frequency in the top-K locations can also be returned and the frequency error can be bounded based on the value of $\Delta$ and as for Lemma~\ref{thm:errorBound}. 
For answering a TFS query \texttt{($t,l$)}, we can return the observed frequency of the given location $l$ if it is stored in the Tsum of term $t$. A question is, how to answer the query if $l$ is not stored in the term Tsum.

\begin{lemma}
Given a Tsum and a location $l$ which is not in the Tsum, let {\u f} be
the smallest frequency in the Tsum. {\u f} gives an upper bound on the frequency of $l$.
\end{lemma}

The lemma provides some comfort in giving an upper bound estimate for TFS queries but the estimate can be too high especially if top locations in the summary are quite frequent. A question is if a more accurate estimate can be obtained. We address this using a probabilistic term distribution model, as presented next. 

\subsection{A probabilistic model of term distribution}\label{locality}

To estimate the frequency of locations that are not stored, we need a model of term distribution based on the information that is stored in the summary. Following the spatial variation model of Backstrom et al.~\cite{spatial:variation},
we assume each term has a {\it center} -- a geographical area where the term is most frequent, and a diffusion -- the extent at which the term is concentrated or dispersed around the center.
This model is quite general and explains many natural phenomena such as the energy of waves and earthquakes in the way they spread and the news about events in the way they happen in a location and the word gets out (see Section~\ref{section:related} for more details).
Based on this model, each term $t$ is characterized by its center,
a constant $C$ which gives the frequency of $t$ at its center and an exponent $\alpha$ which describes how quickly the frequency drops as we move away from the center. $C$ and $\alpha$ are referred to as {\it term focus} and {\it term spread} respectively.
Using this model, the probability that term $t$ is observed at distance $d$ from its center is
\begin{equation} \label{equation:genModel}
p = Cd^{-\alpha}.
\end{equation}
A high value of $\alpha$ for a term means that the term is more local and its usage drops rapidly as we move away from its center whereas a low value of $\alpha$ signals a less local (or a more globally spread) term.
Equation~\ref{equation:genModel} gives a generative model, and clearly other generative models may be used as well. Also the model may be extended to support terms with multiple centers (see Section~\ref{section:possible_improvements}). 

A general method to detect the parameters of a term $t$ is to take a maximum-likelihood approach and find the model parameters based on the observed data.
We assume the center of a term falls in one of the locations where the term is frequent and is stored in the summary. This condition may be relaxed for the cost of an increased search space and the optimization cost.
For a fixed center, it is shown that the log-likelihood as a function of $C$ and $\alpha$ has one local maximum~\cite{spatial:variation}, hence the optimal values of $C$ and $\alpha$ can be obtained using a unimodal multivariate optimization approach.

Let $S$ be the set of all occurrences of term $t$, and $d_i$ be the distance of a location $i$ from the term center.
Then
\begin{equation}\label{equation:likelihood}
f(C, \alpha) = \Sigma_{i\in S}^{}\log Cd_i^{-\alpha} + \Sigma_{i \notin S} \log(1 - Cd_i^{-\alpha})
\end{equation}
is the log of the probability for observing the term in locations in $S$,
as estimated by the model with parameters $C$ and $\alpha$. The optimization
process finds the parameter values that maximize $f(C, \alpha)$.
Figure~\ref{figure:model} shows an example of the optimized model for the term \textit{``bitcoin''} in our dataset of tweets collected from December 1st 2017 to January 1st 2018 in the US. The term center falls in the state of New York.

\begin{figure}
\centering
\includegraphics[width=2.4in,height=1.6in]{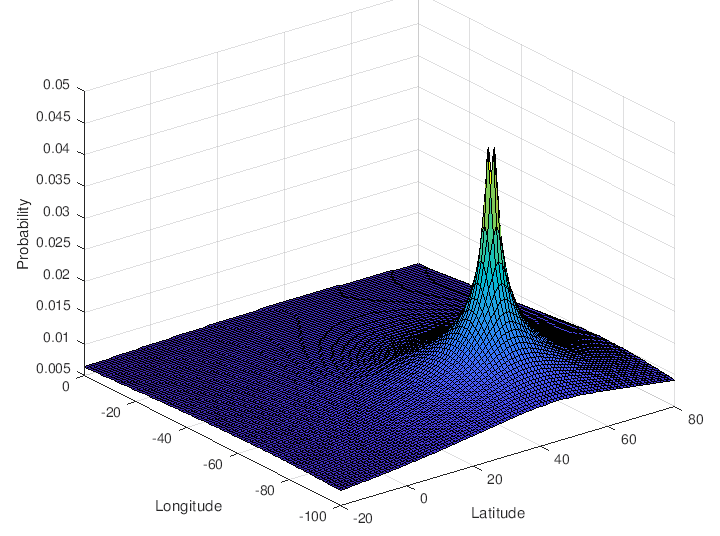}
\caption{Optimized model for term \textit{bitcoin} (Twitter dataset)}
\label{figure:model}
\end{figure}

There are some challenges in using the model in a streaming setting though.
First, solving the optimization problem (as stated above) requires the knowledge of all locations in which the term has appeared as well as all locations where it has not appeared.
Since a Tsum only stores the frequency data for a selected set of locations, a question is if this provides sufficient data to estimate the model
parameters and if the model is accurate.

Second, as the data streams in, the parameters of the model can change. For example, the center can move and the degree of dispersion around the center can vary. For more accurate estimates, the model needs to be updated before answering a TFS query, but running the maximum likelihood for each query can be computationally intensive.
A question is if this can be done efficiently in a streaming setting.
To put the questions together, we want to find ways of balancing the accuracy of the estimates with the efficiency of the results. We address this in the rest of the paper.

With the model parameters estimated for a term $t$, the probability of the term occurrence at any location $l$ of the map can be estimated. Let $p(l)$ denote this probability. Tsum will also give the total term frequency in all locations (see the next lemma). If $f_t$ denote this frequency, then $f_t \times p(l)$ gives the expected frequency of $t$ in location $l$, which can be returned for TFS queries.

\begin{lemma}
\label{lemma:totalFrequency}
Given a Tsum for a term $t$, the sum of all frequency values in the Tsum will give the total frequency of $t$.
\end{lemma}
\begin{proof}
For each occurrence of a term $t$ in the stream, SpaceSaving updates a counter in the term Tsum. This update is either in the form of adding one to the frequency of a location where the term appears (if the new location already exists in the Tsum) or replacing a location with the new location before adding one to the frequency value and updating the error. In both cases, every occurrence of the term is counted, 
and the sum of frequencies gives the total frequency of the term. 
\end{proof}

\noindent
{\bf Tsum+} Consider the scenario where the model parameters are estimated using a Tsum. We refer to this approach as Tsum+, distinguishing it from Tsum which is not using the probabilistic model.
The accuracy of estimates based on Tsum+ is expected to be high when enough counters can be allocated (as shown in our experiments). However the running time of TFS queries may not be acceptable. 
This is because function $f(C, \alpha)$ in Equation~\ref{equation:likelihood} includes a term for every location where the term appears in as well as a term for every location where the term does not appear in. The number of those distance calculations is equal to the number of cells on the map. Also, to find the parameters that optimize $f(C, \alpha)$ in Equation~\ref{equation:likelihood}, one needs to compute the function multiple times. We address this problem in our next approach.

\section{Ring Summary Method}\label{section:ring_summary}

The idea behind the ring summary method (also referred to as Ringsum) is to speed up TFS queries by reducing the number of calculations for cells that are within the same distance of a center. 
This is done by maintaining a summary for locations that have potentials to be centers. In a simple setting, those potential centers are the locations where a desired event or term is frequent. Section~\ref{section:possible_improvements} discusses refinements to the storage of the centers to improve the query processing time and the accuracy.
The summary is expected to maintain the distribution around each center to allow a more efficient evaluation of the queries. 

\noindent
\textbf{Ringsum structure}
Each center in the summary has a Cell Id denoting a location, a term frequency $f$ in the center, an error value $\Delta$, and a set of ring counters, each counting the occurrences of the term in a ring around the center. As in a Tsum+, Cell Id, $f$ and $\Delta$ keep track of the top-K frequent locations for each term.

We introduce the rings to accelerate the process of finding the spread of a term at each of its centers, as well as to increase the accuracy of spread while using smaller summary sizes. The assumption is that since the spread of a term is determined only based on the number of its occurrences at each distance (according to Equation~\ref{equation:likelihood}), we can approximate the occurrence of the terms in every cell on the map except the center by a distance from the center and a frequency which shows the number of times the term shows up in that distance from its center. 
 
The stream processing in the presence of ring counters is slightly different though. 
Instead of storing all the occurrences of a term on the map, these new counters will store a summary of the occurrences of the term around each of its centers. 
Figure~\ref{fig:summary} demonstrates the new summary structure.
\begin{figure}
\centering
\includegraphics[height=1.1in]{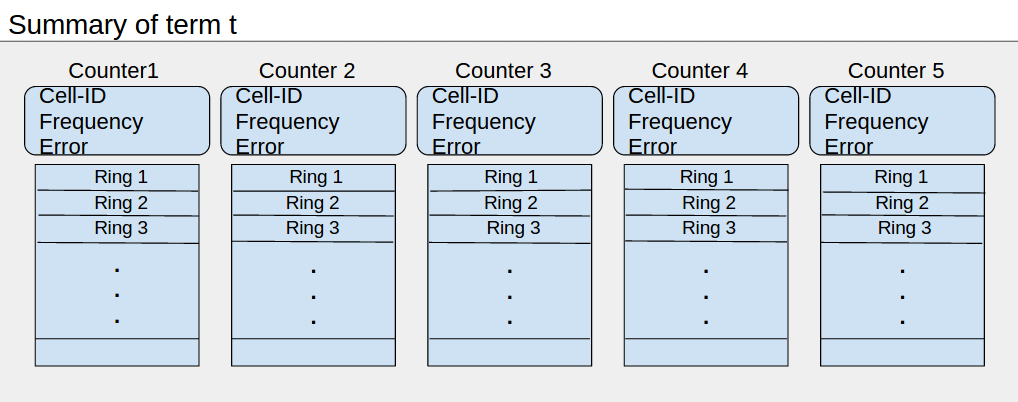}
\caption{The structure of term summaries with ring counters for every center of the term.}
\label{fig:summary}
\end{figure}
The ring counters for a center are updated as follows:
First, we select $R$ distance intervals $\{(0, d_1], (d_1, d_2], (d_2, d_3], \ldots (d_{R-1}, d_{R}]\}$ 
and count the number of occurrences of the term at these distance intervals from its center (see Section~\ref{ring-size} for details of setting distance intervals). 
The ring counters are all initially set to zero.

For every occurrence of the term, if the distance falls in the $(d_{i-1}, d_i]$ interval, the counter for $(d_{i-1}, d_i]$ is incremented. In other words, we save the number of times the term appears in each distance interval from its center. The ring counters provide an approximation of all the occurrences of the term across the geographical map. As the occurrence of a term falls further away from its center, its effect on the value of spread is lessened. Hence a good approximation of the spread may be obtained using a rather small $R$.

There are some challenges when rings are used instead of cell counters.
First, as the data streams in, the center of a term can change. How can the rings be updated when the center of a term changes?
Second, given $R$ rings, we want to place them around a center such that the estimation error is minimal. What is an optimal placement of the rings around a center? These challenges are addressed in the rest of this section.

\subsection{Handling Changes in Potential Centers}
As the data streams in for a term, a top-K location (treated as a potential center) may be swapped with a new location. We need a mechanism to populate the list of ring counters for this new location.

One approach is to reset the rings, discarding all saved frequencies in each ring, and count only the future occurrences of the term around its new center.
This approach, although straightforward, can result in a large error in the estimated spread $\alpha$ of a term. We discuss this alternative further in Section~\ref{section:possible_improvements}.

Another approach is to estimate the term frequency in each ring around a new center based on the information that exists in the summary for other centers.
If we can assume that the distribution of events in each ring around the center is uniform~\footnote{This assumption clearly does not hold for all centers, especially when part of a ring is not habitable due to land constraints (e.g., when the center of a ring is placed on the coast line or near a mountain). One may keep an index of places or cells where the events of interest cannot happen, and take the area of those regions into account when estimating the frequency of a new center.}, we can say that a proportion of the frequency saved for an old set of rings can be transferred to the new set of rings. This proportion is 
the intersection area of the rings around the existing centers and the new center, calculated as the ratio of the area of intersection between the new rings and the old rings to the whole area of the old rings. Having this intersection ratio and the term frequency in the old rings, we can have an initial frequency estimate for each ring of the new center. 
\begin{figure}
\centering
\includegraphics[height=1.5in]{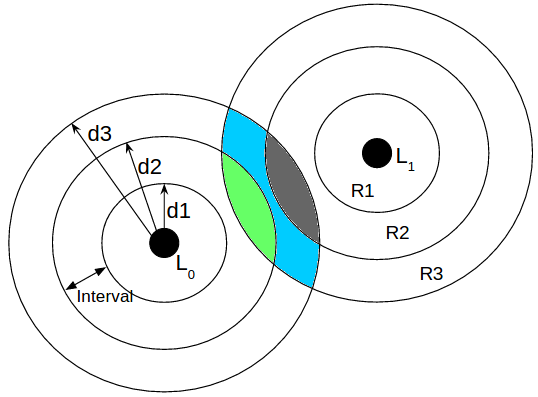}
\caption{Finding the intersection between the new set of rings and the previous set of rings for a replaced center location to initiate the new ring values.}
\label{fig:ringsIntersect}
\end{figure}

Let $L_0$ be an old center and $L_1$ be a new center, and suppose the rings around these centers are as shown in Figure~\ref{fig:ringsIntersect}.
The blue area is the intersection area between R3 from the new center and R3 from the previous center. The grey area is the intersection between R2 from the new center and R3 from the previous center and the green area is the intersection between R2 from the previous center and R3 from the new center. Let's call the grey area A, the blue area B, and the green area C. In this example the initial values for R1, R2 and R3 for the new center $L_1$ will be:
\begin{equation}
\begin{split}
    \phi_{R_{1, 1}} & = 0 \times \phi_{R_{1, 0}}, \\
    \phi_{R_{2, 1}} & = \frac{A}{area(R3)}\times \phi_{R_{3, 0}}, \\
    \phi_{R_{3, 1}} & = \frac{B}{area(R3)} \times \phi_{R_{3, 0}} + \frac{C}{area(R2)} \times \phi_{R_{2, 0}},
\end{split}
\end{equation}
where $\phi_{R_{i, j}}$ is the value stored in the $i$-th ring counter of $L_{j}$, and $area(R2)$ and $area(R3)$ are the areas of respectively the second and the third rings around the old center. The first ring of the new center is initialized to $0$ because there is no intersection between its area and any rings from the previous center.

\begin{algorithm}
\small
\caption{Update(location l, summary S, size m)}\label{ring_update}
\begin{algorithmic}[1]
\If {$l \in S$} 
\State $ S[l].f \gets S[l].f + 1$
\For {i in $\{S - S[l]\}$}
\State $S[i].addToRings(l)$
\EndFor
\Else 
\If {$S.size < m$}
\State $S[l].f \gets 1$
\State $S[l].\Delta = 0$
\For {i in $\{S - S[l]\}$}
\State $S[i].addToRings(l)$
\EndFor
\Else
\State $j \gets findMinCounter()$
\State $S[l].\Delta \gets S[j].f $
\State $S[l].f \gets S[j].f + 1$
\State $S.delete(j)$
\State $S[l].initializeRings(l)$
\For {i in $\{S - S[l]\}$}
\State $S[i].addToRings(l)$
\EndFor
\EndIf
\EndIf
\end{algorithmic}
\end{algorithm}

\subsection{The Update Algorithm}
Algorithm \ref{ring_update} gives the steps for updating a summary in the Ringsum method.
The \textit{addToRings()} function finds the distance between a newly seen location on the stream and the center locations in the summary. If the distance falls in any ring around a center, then the frequency count for that ring is incremented. 
The time cost of finding the ring a newly seen location falls into is $O(R)$ if the rings are sequentially scanned and $O(log R)$ using a binary search. 

The \textit{initializeRings()} function finds the intersection between each previous ring and all new set of rings around a new center and calculates the portion of the frequency from the previous set of rings to be transferred as initial values to the new set of rings. The time cost of this function is $O(R^2)$.
The time cost of the update algorithm in the ring summary method depends on the number of center counters, $m$, and the number of rings per center counter, $R$. In the worst case the update time is $O(mR + R^2)$ where $R$ can be treated as a small constant. As for the memory usage, Ringsum  stores $R$ ring counters per center and has a memory usage of $O(mR)$. 

One caveat in estimating the focus and spread using Ringsum is that the summary keeps a count for all cells within distance $d_{i-1}$ and $d_i$ of a center whereas in Equation~\ref{equation:likelihood} we need the actual distance of each cell from the center. A solution is to use the mean distance, i.e. $d_{approx_{i}} = \frac{d_{i-1} + d_i}{2}$, for all cells at ring $i$.
This reduces the distance calculation for ring $i$ to a single term. 

Another caveat in estimating the spread is that we also need the locations where the term does not appear. To reduce the overhead in processing those locations, we can compute an offline table which keeps the number of cells that fall inside each of the rings around a central cell on the world map. With the world map treated as a sphere and using the great circle distance as the shortest distance between two points on the surface of earth, defined as $r \theta$ where $r$ is the radius of the earth and $\theta$ is the central angle in radians between the two points, the number of cells at each distance of a center cell varies with the location of the center. However, one can keep for each cell the number of cells that fall inside each of the rings around the cell at a space cost of $O(c^2 R)$ for a $c \times c$ grid.

%

Given a center and its ring counters for a term, 
the number of cells at each distance interval from the center where the term does not appear can be obtained by subtracting the ring counter value for that interval from the number of cells in that distance interval (as recorded in the table described above).
Suppose our Ringsum has $R$ rings for a center. Let $d_{exp_{i}}$ denote the expected distance of cells in Ring $i$ to the center and $\phi_{R_i}$ be the number of occurrences of the term in Ring $i$. If $T_i$ denote the number of all cells in Ring $i$ (as pre-computed and stored for easy look up), then $f(C, \alpha)$ can be written as:
\begin{equation}
    \begin{split}
    f (C, \alpha) = \Sigma_{i=1}^{R}\log (\phi_{R_i} C d_{exp_{i}}^{-\alpha}) +\\ \Sigma_{i=1}^{R} \log((T_i-\phi_{R_i})(1 - C d_{exp_{i}}^{-\alpha})).
    \end{split}
    \label{equation:ring}
\end{equation}

All variables in Equation~\ref{equation:ring} are known and the equation has only one term for each ring. The number of terms to compute $f(C, \alpha)$ is now decreased to the number of rings, and this drastically reduces the time cost of TFS queries. 

As discussed in Section~\ref{locality}, the number of distance calculations for computing $f(C, \alpha)$ using Tsum+ is equal to the number of grid cells on the map. Given a $c \times c$ grid system and assuming each center has $R$ rings in the summary, the time cost of computing $f(C, \alpha)$ is $O(R)$ for Ringsum and $O(c^2)$ for Tsum+. 
It can be noted that the time cost for TFS queries only depends on parameters $R$, $c$ and the summary size $m$.
Also, since $R << c^2$, Ringsum is expected to be much faster. 

\subsection{Setting Ring Distance Intervals }  \label{ring-size}
A question that arises in the ring summary method is how the rings should be placed 
since both the number of rings and the distances between them can vary.
The number of rings depends on the space constraint, and with a larger number of rings, the spread can be calculated more accurately because the distances registered for the terms inside each ring will be closer to their real distances.

Selecting the sizes of ring intervals is not trivial. Ideally one needs to know the term distribution on the map to allocate the rings. However, the ring sizes are set before the terms are seen on the stream, and the term distribution is less likely to be known in advance. 
Let $D$ be the radius of the largest possible ring around a center, and $r \times D$ be the radius of the next possible ring around the same center where $ 0 < r < 1$. Assuming that the ratio $r$ remains the same between two consecutive rings, the rings can be
defined by fixing $r$ and the number of rings.

One observation that can be made based on Equation~\ref{equation:ring} is that the occurrences in rings that are farther away from the center have little impact on the value of $f(C, \alpha)$, whereas the occurrences in rings that are closer to the center have a bigger impact. 
Therefore, an optimal ring setting is expected to keep more accurate distances for locations that are closer to the center. That means, smaller intervals should be assigned to the rings that are closer to the center and larger intervals should be given to the rings that are farther away from the center.

\begin{lemma} \label{lem1}
Given a circular region with radius $D$ and assuming that the distribution of points in the region is uniform, $\frac{2}{3}\times D$ is the average distance from the center to a point in the region.
\end{lemma}
\begin{proof}
Consider a unit circle and a set of points that are uniformly distributed in the circle. 
The Probability Density Function (PDF) of 
the points at distance $x$ is 
$2x.\mathbbm{1}_{[0,1]}(x)$
and the expected value of the distance is given by
$\int_{0}^{1}2x^2dx = \frac{2}{3}$.
\end{proof}

Suppose the distribution of locations around a term center is uniform. If we allocate only one ring for a center, $\frac{2}{3}D$ is the average distance between the points in the ring and the center. For more than one ring, we may divide the region at where the average falls and set the ratio $r=\frac{2}{3}$, giving the set of radiuses $\{\frac{2D}{3}, \frac{4D}{9}, \ldots\}$
for the rings. In our experiments, we evaluate this along with other settings of $r$ in terms of their error in calculating the spread.

\section{Multi-term Queries}\label{multi}

Consider queries that have more than one term. It is easy to evaluate them if the exact term frequencies in all locations are known. If $f_1$ and $f_2$ denote the frequencies of terms $t_1$ and $t_2$ respectively in a location, then $min(f_1,f_2)$ gives the frequency of both terms in the same location. However, Tsum+ and Ringsum do not keep the frequency of a term in all locations, and only the frequency of top locations are kept. To estimate the joint frequency of a multi-term query, one may assume that the term frequency in every location that is not kept in Tsum+ or Ringsum is zero. This will clearly have error depending on the summary size. Also if the sets of top locations maintained for two terms in a multi-term query are disjoint, the query results in an empty set.
An alternative approach is to use our probabilistic model for each term (as discussed earlier) to estimate a joint probability of occurrence for multi-term queries. 

Given a multi-term query $T = \{t_1, t_2, ..., t_n\}$, an RFS query may find top $K$ locations where the terms all co-occur and the locations are frequent. With the top-K frequent locations maintained for each term (in their Tsum+ or Ringsum), 
and treating those locations as potential centers, the focus ($C$) and spread ($\alpha$) for each of these centers may be estimated; this gives $K$ probabilistic models $P = Cd^{-\alpha}$ per term. Each model will provide the probability of the term occurring in any cell on the map using $p = C d^{-\alpha}$ where $d$ is the distance from the center to the cell.

The join probability in every location on the map can be computed under some assumption. In particular, if we can assume that the appearance of the terms in a multi-term query is independent from each other, then the joint probability can be estimated as the product of the probabilities, i.e.
\begin{equation}\label{joint}
    P(t_1, t_2, \ldots, t_n) = P(t_1) \times P(t_2) \times ... \times P(t_n).    
\end{equation} 
If the occurrence of the terms in a multi-term query is not independent, then the lowest probability can give an upper bound estimate of the probability, i.e.
\begin{equation}\label{min}
  P(t_1, t_2, \ldots, t_n) \leq  min(P(t_1), P(t_2), \ldots, P(t_n)).  
\end{equation}

With the probabilities of query terms estimated for all locations, an RFS query can return the top K locations with the highest probabilities.  
Algorithm \ref{alg:multi-term} gives the steps of answering an RFS multi-term query with the assumption that the terms are independent.

\begin{algorithm}
\small
\caption{QueryProcess(multi-term T, K)}\label{alg:multi-term}
\begin{algorithmic}[1]
\For {$t \in T$}
\For {$center_i$ in top-K centers of $t$}
\State $P_i \gets C_i\times d^{-\alpha_i} $
\For {cell in Map}
\State $d_{cell} \gets getDistance(center_i, cell)$
\State $p_t(cell) \gets C_i \times d_{cell}^{-\alpha_i}$
\EndFor
\EndFor
\EndFor
\For{$j=1$ to $K$}
\State $JointProbabilities_j \gets \prod_{t, cell}p_t(cell)$
\State top-j result = $Max(JointProbabilities_j)$
\EndFor

\end{algorithmic}
\end{algorithm}


Answering TFS queries is similar except that we find the joint probability of the terms only in the given location. The product of this probability with the minimum total frequency of the terms in the multi-term query will give an upper bound estimate of the expected frequency of the multi-term. 

\section{Improvements and Extensions}\label{section:possible_improvements}

The proposed model and algorithms can be improved or extended in few directions. Our first three improvements are geared toward the ring summary method, which despite being very efficient for RFS queries in terms of the query time, has some drawbacks. First, the stream processing cost is higher in the worst case due to the overhead in replacing an existing center with a new one. Second, for a fixed space usage, the accuracy of focus and spread can be lower than that of Tsum+ especially if not enough center counters can be allocated because of the space allocated for rings.

In this section, we first develop three variations of the ring summary method to address the aforementioned issues. Next we discuss the scenario where a term has more than one center.  

\subsection{Fixed-Center Ring Summary}
To avoid the overhead in stream processing when replacing a center, one strategy is to stop updating the centers after a certain volume of the stream passes.
The intuition behind this is that for some terms, the top central locations do not change after certain fraction of the stream passes, and there is not much point in updating the centers that may not end up in top. Our experiments show that the average number of replacements in top-K counters of the summary decreases noticeably after about 20\% of the stream is received (see Section~\ref{sec:exp-improvements} for more details). Also the complex updating process of the rings in the ring summary method, when their center counters are replaced with a new location, may increase the error of spread. Fixing the center locations can result in lower error for the calculated spread if we know that the centers are not going to change anymore. 

\subsection{Light-Update Ring Summary}
Another strategy to avoid the stream processing overhead due to the changes in centers is to set the initial ring counters to zero when a new center is added. The intuition here is that the existing rings cannot contribute much to the rings of a new center if the area of overlap between the rings is small. Under this condition, the contribution of the existing rings may be ignored to save in processing time. 
This strategy can introduce more error in calculating the spread but it will decrease the summary update time drastically. 

\subsection{Proximity-Aware Update}
The idea behind this strategy is that when a center moves back and forth between a set of adjacent grid cells, any of those cells can be a center. Also for two nearby centers, the occurrence data can be very similar, each leading to the same or similar probabilistic model, which is as good as the other. Hence we may avoid unnecessary updates and only keep track of centers that are within an acceptable distance from each other.
This strategy can be implemented by slightly changing our update algorithm.
When all center counters in the summary are full, a new potential center is inserted only if the new location does not fall, for example, within the first ring of other centers. Reducing the number of updates in centers can decrease the error for spread but the centers may not also be very accurate. 

\subsection{Multiple Centers}
The single-center model is sufficiently accurate in capturing the distribution of many terms, but some terms are better modeled using multiple centers. The paper has focused on the simple single-center setting, however, this is for the convenience of the presentation and not due to a limitation of the model or our approach. Our sketches (both TSum and RingSum) keep the frequency stats for all locations where a term is frequent, and any of those locations may be treated as a center. For example, suppose a term has two centers {\it a} and {\it b}, each with its own focus and spread. The probability of observing the term is now defined in terms of two functions, say $p_a$ and $p_b$ (each defined as in Equation~\ref{equation:genModel}). This probability at an arbitrary location of distances $d_a$ and $d_b$ from the two centers can be computed as $p_a(d_a)+p_b(d_b)-p_a(d_a)p_b(d_b)$ assuming independence. This quantity can now be plugged into Equation~\ref{equation:likelihood} to find the model parameters. 

Increasing the number of centers, however, adds overhead to the estimation process. One issue is deciding on the number of centers, which can be subjective
and has similar characteristics as detecting the number of
clusters~\cite{he2010detecting,evanno2005detecting}. In some cases, the number of centers may be given as a parameter (similar to the number of clusters that may be provided in a clustering approach). The number of centers may also be detected on the fly as the data streams in. For example, one may start with one center and monitor the estimation error. When the error reaches a threshold, increasing the number of centers may reduce the error. Increasing the number of centers also adds overhead to the maximum-likelihood optimization. Hence for large number of centers, a clustering-based heuristic may be used~\cite{kumar1999trawling,jeub2015think}. For example, a model may be fit to each center locally before each geographic point on the map is assigned to a center which gives that point the highest probability. 

A deeper study of our model in the presence of multiple centers and different update strategies and exploring areas where the model may be optimized or our query evaluation strategies may be improved is an interesting future direction. 

\section{Experiments}\label{section:experiments}

To evaluate the efficiency and the effectiveness of our algorithms and their robustness to our parameter settings, we conducted a set of experiments. The objective was to evaluate our probabilistic model, built on a stream summary, and the performance of our algorithms and strategies. In particular, we evaluated
(1) the accuracy and robustness of our probabilistic model under different parameter settings (Section~\ref{exp:prob-model-eval}) and ring sizes (Section~\ref{size-selection}); 
(2) the accuracy of single-term and multi-term queries (Sections~\ref{exp:eval-RFS}, \ref{exp:tfs}, and \ref{exp:eval-multiterm}); (3) running time of queries and updates (Section~\ref{exp:eval-runtime}); and
(4) the effects of improvements discussed in Section~\ref{section:possible_improvements} on running time and accuracy (Section~\ref{sec:exp-improvements}).

\subsection{Experimental Setup}
Our experiments were conducted on two datasets: (1) Flickr YFCC100m dataset~\cite{yfccm} and (2) a collection of tweets collected from Twitter.
The Flickr dataset included 100 million photos, of which roughly 60\% had geo-tags. We used terms from the description field and user and machine tags of geo-tagged photos to create our first dataset.
The second dataset included a set of 
53 million geo-tagged tweets, posted from within the US and gathered from Twitter's stream API \cite{streamAPI} from December 1st 2017 to January 1st 2018.
For each dataset, we used a random query set of 1000 terms that were at least 0.02-frequent (i.e. appeared in at least 2\% of the stream). The query terms were randomly chosen from the bag of all terms in the stream, meaning that more frequent terms were more likely to be picked. 
All experiments are done on a machine running Ubuntu 16.04 with 8GB of memory and core i7 Intel CPU, and all our code is available online~\footnote{Our code is available at \url{https://github.com/Sara-Farazi/Top-K-Location-Query}.}.

\subsection{Ring Size Selection} \label{size-selection}
Our first experiment is on the ring size selection and the robustness of our term distribution model to the ring size. It is expected that the more the number of the rings, the more accurate the count estimates and the less the error in estimating the spread. However, having a fixed number of rings, we wanted to know what the ring sizes should be. 
For this experiment, we varied the size ratio $r$ from $0.1$ to $0.9$. 
This gave us a sequence of rings with the radius of each ring in the sequence reduced by a factor of $r$ to give the radius of the next ring.
Since the smallest ring included a $1^\circ$x$1^\circ$ latitude/longitude grid cell and the largest ring included all the cells, the number of rings varied for each $r$.
For example, there were 3 rings for $r=0.1$, 4 rings for $r=0.2$ and 10 rings for
 $r=0.6$; the number of rings were capped at 10 for larger values of $r$.
For our ground truth, the parameters of the probabilistic model were estimated using all occurrences of the terms.
The relative error in spread $\alpha$ are shown in Figure~\ref{fig:ring-size}, where the mean and standard deviation of error are computed for 1000 queries in each dataset. At $r=0.6$, the Flickr dataset has its minimum mean error and
the Twitter dataset is not that far from the minimum. This value of $r$ also is not that far from the ratio $\frac{2}{3}$ discussed in Section~\ref{section:ring_summary}, which means the ring size for both datasets can be set based on our analytical finding, and independent of the dataset being used.
Unless explicitly stated otherwise, $r=0.6$ for the rest of our experiments.
\begin{figure}
\centering
\includegraphics[width=3in,height=2in]{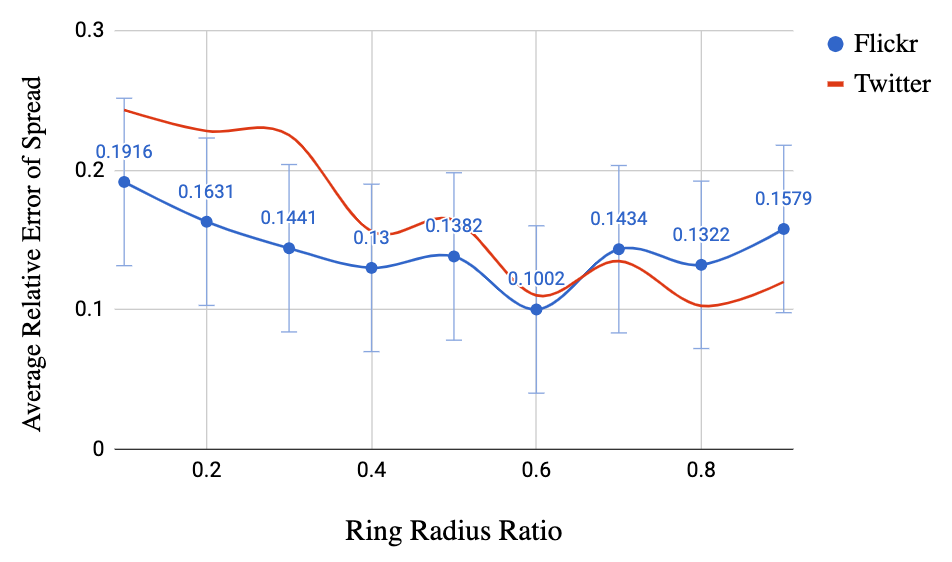}
\caption{Error of spread varying the ring size in Ringsum}
\label{fig:ring-size}
\end{figure}

\subsection{Probabilistic Model Evaluation}
\label{exp:prob-model-eval}
Both TFS and multi-term RFS queries estimate the frequency of a term in an arbitrary location using our probabilistic model, hence their accuracy largely depend on the accuracy of the probabilistic model. In this section we evaluate the probabilistic model under different parameter settings. 
We allocated the same space for both Tsum+ and Ringsum and
compared their relative error of focus and spread and the accuracy of centers. 
For each center, a location id, a term frequency in the location and an error was kept whereas for each ring, a counter was sufficient.
Both the error and the accuracy were measured compared to the case where the model parameters were estimated using all occurrences of the terms. 
Table~\ref{error} shows the results
of this experiment on the Twitter data. The first line with zero rings represents Tsum+, and the other lines show the performance of Ringsum under different parameter settings.
Overall, in our random set of 1000 queries, Tsum+ has less error in focus and spread than Ringsum. Increasing the number of rings and reducing the number of centers reduces the error of focus and spread for Ringsum, for the cost of a reduced center accuracy. Note that the memory size is fixed for a fair comparison with Tsum+, and increasing the ring size entails fewer center counters. 

The same experiment on the Flickr dataset gives similar results (and for the same reason not reported), though the error of spread is relatively lower. This is mainly because terms in the Flickr dataset are more focused around their centers and our method finds more accurate models for such terms. 


\begin{table}[htb] 
\small
\centering
\caption{Accuracy of the probabilistic model for different settings of center counters and rings on the Twitter data}
\begin{tabular}{|c|c|c|c|c|c|} \hline
&&center&\multicolumn{2}{|c|}{relative error} & center\\
$r$&rings&counters&C&$\alpha$ & accuracy\\ \hline
- & 0& 216 &0.161 & 0.028 & 96.7\\ \hline
0.3& 5& 81& 0.215 & 0.203 & 94.1 \\
0.4& 7& 65& 0.188 & 0.138 & 92.2 \\ 
0.5& 8& 59& 0.167 & 0.145 & 91.8\\
0.6& 10& 50& 0.158 & 0.090 & 91.4\\ 
\hline\end{tabular}
\label{error}
\end{table}



\subsection{Accuracy of Single-Term RFS Queries}
\label{exp:eval-RFS}
The accuracy of single-term RFS queries is the same as that of Tsum, our baseline method for which evaluation is reported elsewhere~\cite{FaraziRafiei19a, space_saving}.
Knowing the error threshold $\epsilon$, we can define the number of center counters $m = \frac{1}{\epsilon}$ to capture all $\epsilon$-frequent locations for each term in the summary. As an example, to find all locations that appear $2$\% or more in the stream, one will need at least $50$ center counters, to have all the top locations with accuracy $100$\%. The accuracy of our multi-term RFS queries is reported in Section~\ref{exp:eval-multiterm}.


\subsection{Accuracy of Single-Term TFS Queries}
\label{exp:tfs}
In our next experiment, we employ our term distribution model to answer TFS queries. A question here is how good the model can predict the frequency of a term in an arbitrary location. It should be noted that this is a challenging task when
the frequency is not kept for all locations. For this experiment, we selected 100 query terms randomly and for each term, we further randomly selected 10 locations where the term had a non-zero frequency. The result was a set of 1000 term-location pairs or queries. The ground truth for each pair was the real frequency of the query term in the location, which ranged between 1 and 2000. Our experiment shows that the log of frequency estimates are not that far from that of the actual frequencies. The mean difference for Tsum+ and Ringsum is 0.22 and 0.39 respectively on the Twitter dataset and 0.18 and 0.25 on the Flickr dataset (see Table~\ref{time}). We measure the difference in the log scale to show the error in terms of the order of magnitude difference between estimates and actual values.
As a baseline for comparison, we also evaluated two variations of Tsum without using the probabilistic model. In one variation, the frequency of a query term in locations that were not stored in the summary was treated as zero, whereas in the other variation, the average frequency of the surrounding locations were used. 
The mean (standard deviation of) error in frequency estimates using these simple baselines were 37.4(37.4) and 29.7(32.6) respectively, compared to 16.9(19.6) for Ringsum and 15.5(17.4) for Tsum+.
Our probabilistic model, despite the error in some cases, performs better than our simple baselines for randomly selected query terms.

\subsection{Multi-Term query Performance}
\label{exp:eval-multiterm}
One more area where the probabilistic model is expected to do better is multi-term queries. To evaluate this hypothesis and to assess the performance of our methods, we conducted an experiment to measure the accuracy of top-k frequent locations for a set of multi-term queries. These queries were selected from the set of all two-terms that appeared at least once together in a location. We divided the multi-terms into four bins based on their frequencies:
(1) {\it Rare} included multi-terms that were rarely frequent or appeared together in a location. For example, ``fun girl'' was a rare multi-term.  
(2) {\it Medium} included multi-terms that were more likely to appear together in same locations, such as ``metal music.''
(3) {\it Frequent} included multi-terms that appeared together most of the time in same locations, such as ``high school'' and ``north Carolina.'' 
(4) {\it Very frequent} included multi-terms that were frequent in the stream and appeared together almost all the time. ``New York,'' and ``National Park'' are such multi-terms.

We randomly selected 100 multi-terms from each of the above 4 groups. Our evaluation included 5 algorithms, two were variations of Tsum+ and two were variations of Ringsum. One variation of both Tsum+ and Ringsum was under the assumption that the terms in a multi-term query are independent. Under this assumption, the joint probability of occurrences of a multi-term could be estimated as the product of the individual term probabilities. Another variation of both Tsum+ and Ringsum was under the assumption that the terms in a multi-term query are not independent. Under this assumption, the (upper-bound) joint probability was estimated as the minimum of the individual term probabilities. 
We also constructed a variation of Tsum as a baseline that did not use the probabilistic model. Without using the model, the only way to find the frequency of a multi-term was to take the intersection of individual Tsums 
before selecting top-k location where multi-terms were frequent.

\begin{figure*}[tb]
\centering
\begin{minipage}{.3\textwidth}
\hspace*{-0.5cm}
\includegraphics[width=2.35in]{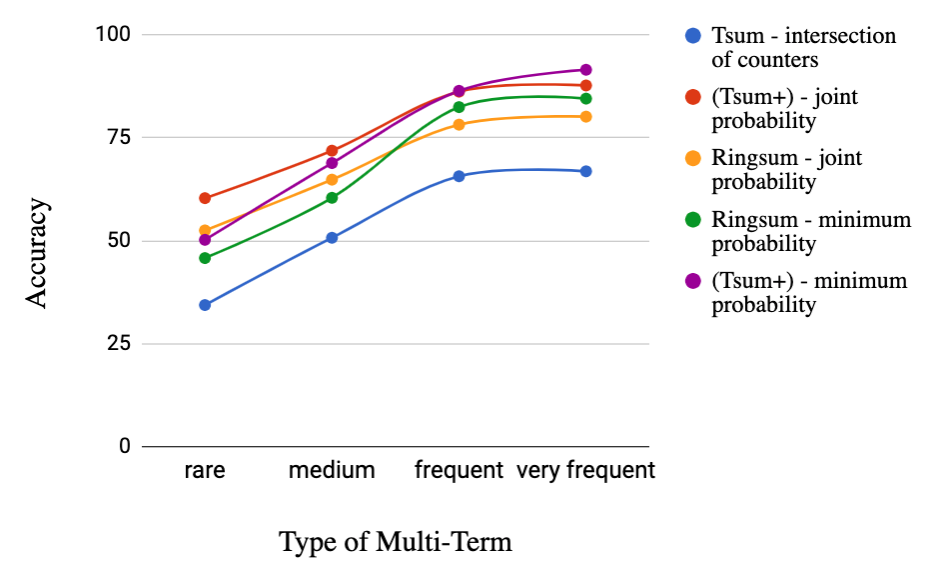}
\caption{Average accuracy of the top location for multi-term queries (RFS) on Twitter data}
\label{multi-term}
\end{minipage}
\hspace*{0.3cm}
\begin{minipage}{.3\textwidth}
\hspace*{-0.5cm}
\includegraphics[width=2.35in]{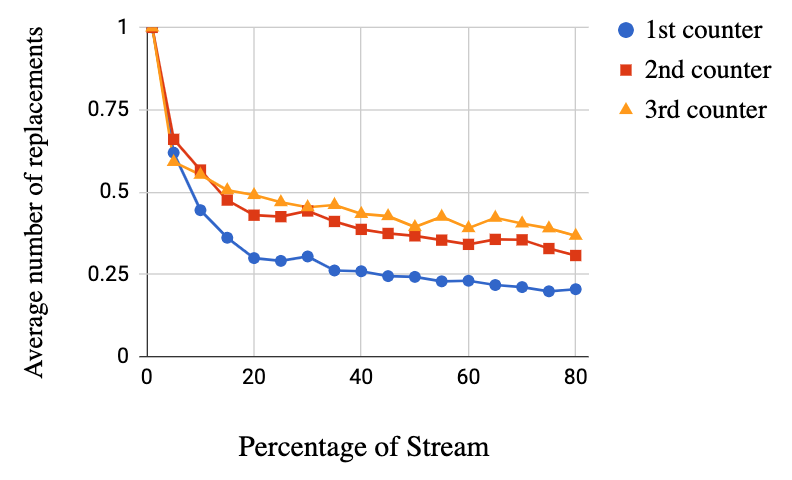}
\caption{The average number of replacements in the first top-location (center counters) of the summaries as more data is received}
\label{variation1}
\end{minipage}
\hspace*{0.3cm}
\begin{minipage}{.3\textwidth}
\hspace*{-0.5cm}
\includegraphics[width=2.35in]{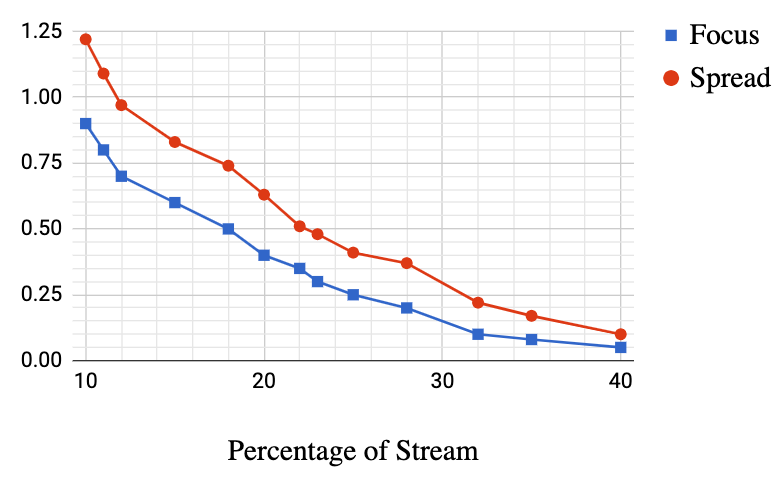}
\caption{The fraction of the stream that must pass (X axis) such that at least half of the centers remain the same for different value of focus and spread (Y axis)}
\label{alpha-c-rate}
\end{minipage}
\end{figure*}

Figure~\ref{multi-term} shows the accuracy of RFS queries for all 5 methods on the Twitter dataset, with queries selecting a top location for each multi-term. A few observations can be made here. First, we can see that the method that does not use the probabilistic model (Tsum intersection) has the least accuracy among all our methods. This is expected since top locations for a multi-term does not necessarily appear in the summary of each of the terms.
In other words, it pays off to use the probabilistic model since it allows us estimate the frequency in locations that are not kept in the summary.
Second, Tsum+ has a slight edge over Ringsum in terms of the accuracy of the results and this is consistent with our result on estimating the frequency of single-terms in a location (as discussed in Section~\ref{exp:tfs}).
Third, the independence assumption seems to work better for multi-terms that are rare or not frequent whereas a non-independent assumption gives more accurate results for multi-terms that are frequent or very frequent.
Finally the accuracy improves as queries become more frequent in a location. This is because such locations are more likely to be kept in individual term summaries, allowing more accurate estimates. 



\subsection{Runtime}
\label{exp:eval-runtime}
In our next experiment, we evaluate Tsum+ and Ringsum in terms of their running times. For each method, we measure the time it takes to evaluate queries and to update the summary as the data streams in.
Our query time is the time to answer a TFS query (RFS queries are easier to evaluate and are not considered here).
Our update time is the time it takes to create and/or update the summary per one tweet (in Twitter dataset) or one photo (in Flickr dataset).
Our experiment shows that Ringsum is two orders of magnitude faster than that of Tsum+ in terms of query time whereas the update time for Tsum+ is faster (see Table~\ref{time}). Both Tsum+ and Ringsum build a term distribution model before answering a query. However, Tsum+ uses all locations in the geographical grid system to build its model, whereas Ringsum only uses the term frequency in the rings. The number of rings is much smaller than the number of grid cells and this explains the difference in query time.
On the other hand, there is overhead in updating the rings especially when a center moves and that is the reason updates take longer time using Ringsum. We next show that the update time of Ringsum can be reduced to the same level as that of Tsum+, using our improvement strategies, while keeping the query time still low.

\begin{table} 
\small
\centering
\caption{Performance on TFS queries on Twitter (Flickr) datasets }
\begin{tabular}{|c|c|c|c|c|c|} \hline
Method &Update     &Query     & Log \\
       &Time (msec)&Time (sec)& Dist. Error\\ \hline
Tsum+& 0.068 (0.097) & 38.15 (22.46) &0.22 (0.18)\\ 
Ringsum &0.736 (0.853) & 0.23 (0.12) & 0.39 (0.25)\\ 
Light-Update&0.161 (0.242) & 0.23 (0.12) & 0.54 (0.37)\\ 
Fixed-Center&0.014 (0.018) &0.23 (0.12)& 0.34 (0.22)\\
Proximity-Aware & 0.804 (0.915) &0.23 (0.12)& 0.58 (0.37)\\
\hline \end{tabular}
\label{time}
\end{table}

\subsection{Evaluating Improvements}
\label{sec:exp-improvements}
In this set of experiments, we study the extent at which the improvements discussed in Section~\ref{section:possible_improvements} can boost the performance of the Ringsum method. As shown in our earlier experiments, Ringsum is very efficient in estimating term frequencies in locations and has an accuracy that is close to that of Tsum+. The improvements evaluated in this section target the update time of Ringsum, pushing it closer to (if not better than) the update time of
Tsum+.

Our first improvement looks at the possibility of fixing centers to save time during updates. Figure~\ref{variation1} shows the average number of replacements in the top-3 center counters as the volume of stream that is processed varies for the Flickr dataset. The more
volume of the stream is received, the fewer replacements are expected.
For example, after seeing about 40\% of the stream, less than 25\% of the centers are expected to be replaced. Based on this observation, a variation of Ringsum is to fix the centers after some fraction of the stream passes. In our experiments, fixing the centers after seeing 20\% of the stream not only reduced the update time significantly (as shown in Table~\ref{time}) but also led to a slightly less error.
In another experiment, to understand the relationship between the changes in centers for different values of focus and spread, we set to find the fraction of the stream that must be seen until at least half of the centers remain unchanged.
As shown in Figure~\ref{alpha-c-rate} for different value of $C$ and $\alpha$, 
events with larger spread require scanning a bigger portion of the stream before the centers become stable. This is because events with larger spread generally include more global events that may not be easily tied to a center. A similar pattern can be seen for focus because terms with a larger spread are generally expected to have lower frequencies at their centers.
For $C > 0.4$ and $\alpha > 0.6$, for example, at least 50\% of the centers are expected to be the same after passing 20\% of the stream. On the other hand,
fixing the centers after seeing 20\% of the stream significantly reduces the
relative error of $C$ and $\alpha$ in the probabilistic model (as shown in Table~\ref{improvements}) though the center accuracy can drop when some centers cannot be found.

We also evaluated our two other update strategies, namely light update and proximity-aware update.
Table~\ref{time} gives the time for different update strategies. We can see that light-update can achieve a stream processing time close to that of Tsum+.
Table~\ref{improvements} further shows that the relative error of $C$ and $\alpha$ for light update is more than that of Ringsum, but the update time is much faster.
We can also see that proximity-aware update reduces the error of $\alpha$ due to the less number of replacements in centers. However, center accuracy drops and we can no longer guarantee that all $\epsilon$-frequent locations can be found.

\begin{table} 
\small
\centering
\caption{Error of focus, spread and center accuracy for improved alternative methods of ring summary }
\begin{tabular}{|c|c|c|c|} \hline
Alternate Method & Error of C & Error of $\alpha$ & Center Accuracy\\ \hline
 Fixed-Center & 0.160 & 0.055& 81.3\\ \hline
 Light-Update & 0.204 & 0.122 & 91.4\\ \hline
 Proximity-Aware & 0.175 & 0.082 & 68.2 \\ \hline
 Ringsum & 0.158 & 0.090 & 91.4 \\
\hline \end{tabular}
\label{improvements}
\end{table}

\section{Other Related Works}\label{section:related}

In addition to those reviewed earlier, the literature related to our work can be grouped into 
(1) spatio-temporal frequent term/trend queries.
(2) spatial distribution models of terms, and
(3) finding most frequent items in a data stream.

\subsection{Spatio-Temporal Frequent Term Queries}
Spatio-temporal queries have been long studied (e.g., \cite{samet1990design,zhang2003location}), and more recent work has explored the relationship between spatio-temporal and term queries (e.g., \cite{mahmood2019scalable,de2008keyword}). There have been a few studies on finding top-k frequent or trending terms in a spatio-temporal region~\cite{blogscope,budak,garnet}. For example, Geoscope~\cite{budak2} gets a pair of location and term as input and determines whether the term is trending in the location within the current time window.
Skovsgaard et al.~\cite{jensen:topk} use
a multi-layer grid cell index for storing top frequent terms in different time and space granularities. An extended version of SpaceSaving algorithm is used to keep track of the frequent terms in each cell.
Magdy et al.~\cite{venus} store
the terms in a pyramid index in memory and answer top-k trend queries in a recent time interval, using different trending functions. 
Garnet~\cite{garnet} uses a multi-layer grid cell for spatial indexes and a pyramid tree for temporal indexes, initially stored in main memory and periodically flushed to disk. Garnet partially supports reverse frequent term queries, by scanning through all possible grid cells and returning, if exists, the location that is most frequent for the given keyword. Clearly this approach is not efficient or scalable. 
Ahmed et al.~\cite{ahmed} use a balanced R-tree~\cite{rtree} index augmented by top-k sorted term lists, to support exact answers for top-k spatio-temporal range queries.
There have been also work on finding proximate objects or events, where given a term and a location, one wants to find other events that are similar and not that far from the given event~\cite{cao2012spatial, chen2013spatial}. 

The relevant literature also includes the work on reverse spatial queries (e.g., \cite{achtert2006efficient,lu2011reverse,gao2014efficient,vlachou2010reverse}), where objects that have the query object in their neighbourhoods are retrieved. As noted, the term ``reverse'' in this line of work has a slightly different meaning than that in our case.


As briefly reviewed, different forms of spatial indexes are developed for efficiently finding terms in a location. However, these indexes are not efficient when searching for locations of a given event. They either store all locations and frequencies for a term or ignore the frequency estimation for locations that are not saved. 
To the best of our knowledge, the scalability and the efficiency of reverse term frequency spatial queries are not addressed in the literature. Also most existing works do not support or are not applicable to multi-term queries. 

\subsection{Spatial Distribution Models}
Modeling the geographical distribution of various events and phenomena has been studied in different domains. The model used in our work, where each event is marked with a center and a spread, is quite general and explains many natural phenomena.
In earth sciences, studies of earthquakes present a model for the energy of waves that starts from a center and decays logarithmically as the waves spread further away from the center~\cite{earthquake1}. 
Geographical distribution of species is modeled with exponential decay functions (such as Poisson, Gamma, etc.), which show the decay in population as the distance from their home habitat increases~\cite{species}. 
In human and social studies, the amount of human mobility and movement patterns~\cite{mobility}, population density~\cite{population}, economic growth~\cite{economy}, traffic~\cite{traffic}, etc. are all described with the same model, in which the population has a high density in downtown and commercial centers of urban areas, and as we move away, the density decays exponentially. 
%
%
%
In social media, news about important events are generally focused on the location of the event. 
Spatial distribution modeling has been recently used in the context of terms and events extracted from the web.
Language models~\cite{serdyukov2009, kinsella2011} and generative models~\cite{yu2016geotagging, hong2012} have been used recently to assign a geo-center and a spread for terms used in web content. 
Finally, the probabilistic model of Backstorm et. al. \cite{spatial:variation} quantifies spatial variations of terms in Yahoo search engine query logs, indicating whether a term has a highly local interest or a broader regional or national appeal.


\subsection{Finding Most Frequent Items in a Data Stream}
This problem has been well-studied and two main approaches have emerged: (a) counter-based, and (b) sketch-based.
In a counter-based approach (e.g., \cite{demaine, karp,manku,space_saving}), items in a stream are counted using a fixed number of counters. When all counters are used, new items may be ignored or may be swapped with a currently monitored item.
Sketch-based approaches (e.g., \cite{charikar,cormode,cormode2}) estimate the frequency of items in a stream through a hashing mechanism, where items are mapped to a small set of counters, which are incremented with each occurrence of an item that is hashed to those counters. 
Counter-based techniques are more commonly used for point queries because of their relative error guarantees (e.g., \cite{jensen:topk,garnet}).

\section{Conclusions and Future Work}\label{section:conclusion}
We propose an efficient and scalable method for evaluating reverse spatial term queries in a stream of geo-tagged events. Our method is quite general, allowing the frequency of a given term at any given location to be estimated.
We extend the frequent item counting sketch Tsum with a probabilistic model that allows both TFS queries and multi-term RFS queries to be efficiently answered. We also develop the Ringsum method, which further improves the efficiency of our queries by orders of magnitude.
Our evaluation demonstrates that each of our methods shine in one or more areas of the accuracy of the results, the efficiency of queries, and the efficiency of updates, and our improvements can further boost the performance of Ringsum.

Our work can be extended or improved in a few interesting  directions.
As one possible direction, our parameter estimation, query processing and update algorithms may be studied over sliding windows with old data being removed as the time progresses. Another area for future research is managing the number of terms or events to be monitored. Maintaining a summary for every possible term or event is costly and selecting a subset of terms in advance for some domains may not be trivial. Exploring different grouping strategies and the areas where our summary structure may be improved to better support those groups is an interesting future direction.

Finally, we discussed a few variations of the ring summary method in Section~\ref{section:possible_improvements}, each improving over the base method under different settings. One may consider other variations of the method to account for different event types and term distributions when this
information is available in advance. Consider the scenario where
the center of an event moves over time. For example, a storm can
move from one state to another, and a disease may spread into new
regions. With some knowledge about such changes (e.g., an estimate
of the speed at which an event moves), a question is if better
update strategies can be devised.

\bibliographystyle{ACM-Reference-Format}
\bibliography{references.bib}

\end{document}